\documentclass[onecolumn]{article}
\usepackage[left=1in, right=1in, top=1in, bottom=1in]{geometry}

\usepackage{amsthm}
\usepackage{amssymb}
\usepackage{amsmath}
\usepackage{mathtools}
\usepackage{csquotes}
\usepackage{enumitem}
\usepackage{svg}
\usepackage{hyperref}
\usepackage{xspace}
\usepackage{xcolor}
\usepackage{fancyhdr}
\usepackage{complexity}

\setlist[itemize]{leftmargin=*}
\setlist{nosep}
\setlist{nolistsep}

\hypersetup{
    colorlinks,
    linkcolor={red!50!black},
    citecolor={blue!100!black},
    urlcolor={black!80!black}
}

\newcommand{\defn}[1]{{\textit{\textbf{\boldmath #1}}}\xspace}
\renewcommand{\paragraph}[1]{\vspace{0.04in}\noindent{\bf \boldmath #1.}}

\usepackage{bbm}

\newcommand{\interior}[1]{ {\kern0pt#1}^{\mathrm{o}} }

\newcommand{\set}[1]{\left\{ #1\right\}}

\newcommand{\Z}{\mathbb{Z}}

\usepackage{algorithm}
\usepackage[noend]{algpseudocode} 
\makeatletter
\newcounter{HALG@line}
\renewcommand{\theHALG@line}{\thealgorithm.\arabic{ALG@line}}
\makeatother

\usepackage[capitalise,nameinlink,noabbrev]{cleveref}
\crefname{equation}{}{} 
\crefname{enumi}{Step}{} 

\theoremstyle{definition}
\newtheorem{fact}{Fact}
\newtheorem{defin}{Definition}
\newtheorem{rmk}{Remark}

\newtheorem{prop}{Proposition}

\newtheorem{lem}{Lemma}
\newenvironment{lemma}{\begin{lem}}{\end{lem}}
\newtheorem{clm}{Claim}

\newtheorem{cor}{Corollary}

\newtheorem{thm}{Theorem}
\newenvironment{theorem}{\begin{thm}}{\end{thm}}

\makeatletter
\newtheorem*{rep@theorem}{\rep@title}
\newcommand{\newreptheorem}[2]{%
    \newenvironment{rep#1}[1]{%
        \def\rep@title{\cref{##1}}%
        \begin{rep@theorem}}%
            {\end{rep@theorem}}}
\makeatother
\newreptheorem{theorem}{Theorem}

\usepackage{mdframed}
\usepackage{framed}

\usepackage[
    backend=biber,
    style=alphabetic,
]{biblatex}
\newlang{\HAM}{HAM}
\newlang{\hHAM}{\text{-}HAM}

\title{Complexity of Multiple-Hamiltonicity in Graphs of Bounded Degree}
\author{Brian Liu, Nathan Sheffield, Alek Westover}
\date{}
\begin{document}
\maketitle

\begin{abstract}
We study the following generalization of the Hamiltonian cycle problem:
Given integers $a,b$ and graph $G$, does there exist a closed walk in $G$ that visits every vertex at least $a$ times and at most $b$ times? Equivalently, does there exist a connected $[2a,2b]$ factor of $2b \cdot G$ with all degrees even? This problem is $\NP$-hard for any constants $1 \leq a \leq b$. However, the graphs produced by known reductions have maximum degree growing linearly in $b$. The case $a = b = 1 $ --- i.e. Hamiltonicity --- remains \NP-hard even in $3$-regular graphs; a natural question is whether this is true for other $a$, $b$. 

In this work, we study which $a, b$ permit polynomial time algorithms and which lead to $\NP$-hardness in graphs with constrained degrees. We give tight characterizations for regular graphs and graphs of bounded max-degree, both directed and undirected.
\end{abstract}

\section{Introduction}
\subsection{Background and Previous Work}
The Hamiltonian cycle problem asks, given a graph $G$, whether $G$ admits a closed walk visiting every vertex exactly once. 
A classic result of Garey, Johnson and Tarjan is that Hamiltonicity is \NP-hard even in graphs of max degree 3, and in fact even in 3-regular graphs \cite{GJ3}.

In 1990, Broersma and G{\"o}bel proposed two variants of the Hamiltonian cycle problem: for fixed $k$, is there a closed walk that visits each vertex exactly $k$ times, or at least once and at most $k$ times, respectively? They showed that for constant $k$, both of these problems are \NP-hard \cite{broersmakwalks}. Jackson and Wormald subsequently proved that both problems remain \NP-hard in $j$-connected graphs for any constant $j$ \cite{wormaldkwalks}. These proofs extend to give hardness of the following generalization:

\begin{defin} 
    For fixed integers $1 \leq a \leq b$, a graph $G$ is called $[a,b]\hHAM$ if there exists a closed walk on $G$ that visits each vertex at least $a$ times and at most $b$ times.
\end{defin}

Being $[a,b]\hHAM$ can be considered as a graph factor problem. Letting $2b \cdot G$ be the multigraph obtained by making $2b$ copies of each edge in $G$, $G$ is $[a,b]\hHAM$ if and only if $2b \cdot G$ contains a \defn{spanning even $[2a, 2b]$ factor}; that is, a connected subgraph containing all vertices of $G$, with every degree even and in $[2a, 2b]$. The equivalence of this formulation follows from taking an Euler tour. 

There is a substantial body of combinatorics work seeking sufficient conditions for the existence of graph factors and walks; see Plummer's survey \cite{plummersurvey}, the book by Akiyama and Kano \cite{akiyamakanobook}, or the survey by Kouider and Vestergaard which focuses specifically on connected factors \cite{kouidersurvey}. For example, Gao and Richter showed that any $3$-connected planar graph is $[1,2]\hHAM$ \cite{2walk3connectedplanar}, and Schmid and Schmidt gave a polynomial-time algorithm to find such a walk in which only the 3-separators are visited twice \cite{2walk3connectedplanarpolytime}. Several variants of the problem have also been consider from an algorithmic complexity standpoint. Ganian et al consider the case where $a$ is nonconstant in $n$, showing that $[a,a]\hHAM$ is $\NP$-hard for $a \leq O(n^{1-\varepsilon})$, but in $\RP$ for $a \geq \Omega(n/\log n)$ \cite{bigaa}. Nishiyama et al show that there is a polynomial-time algorithm to determine whether a graph has a walk visiting every vertex an odd number of times, but that the problem becomes $\NP$-hard when each edge can be traversed at most $3$ times \cite{pairityham}. 

In this work, we will study the complexity of $[a,b]\hHAM$ in regular graphs and graphs of constant maximum degree. Prior papers have investigated degree conditions sufficient to guarantee $[a,b]\hHAM$ --- for instance Kouider and Vestergaard showed that, for $a \geq 2$, any 2-connected graph with minimum degree $an/(a+b)$ is $[a,b]\hHAM$ \cite{evenabmindeb}, and independent papers by Gao and Richter, Jin and Li respectively showed that any $2$-connected graph with maximum degree $2b - 2$ is $[1,b]\hHAM$ \cite{2connectedmaxdeg, 2connectedmaxdeg2}. However, to our knowledge we are the first to consider the algorithmic complexity of $[a,b]\hHAM$ in general bounded degree graphs.

\subsection{Results}

Our main results are tight characterizations for the hardness of $[a,b]\hHAM$ in regular graphs, graphs of bounded max-degree, and bounded max-degree digraphs. Interestingly, the thresholds for which $a,b$ lead to $\NP$-hardness are very different in each of these settings. In regular graphs, we find that hardness is determined by the maximum visit count $b$ in comparison to the degree $d$, where the parity of $d$ is very important:

\begin{reptheorem}{regular}
    When restricted to $d$-regular graphs, $[a,b]\hHAM$ is \NP-hard if either 
    $d$ is even and $b < d/2$, or 
    $d$ is odd and $b < d$.    
    Otherwise, $[a,b]\hHAM \in \P$.
\end{reptheorem}

In graphs of bounded max degree, we find instead that hardness is determined by the ratio of $b$ to $a$ in comparison to the max degree $d$, and that $a = 1$ is a special case:

\begin{reptheorem}{maxdeg}
    When restricted to graphs of maximum degree $d$, $[a,b]\hHAM$ is \NP-hard if either 
    $a = 1$ and $b < d$, or
    $a>1$ and $\frac{b}{a} < d-1$.
    Otherwise, $[a,b]\hHAM \in \P$.
\end{reptheorem}

For directed graphs, $[a,b]\hHAM$ turns out to always be hard in graphs of maximum degree $4$, and hard in graphs of maximum degree $3$ unless $a = b$ or $b = 1$:

\begin{reptheorem}{3digraph}
    In directed graphs of maximum degree $3$, for $1 \leq a \leq b$, $[a,b]\hHAM \in \P$ if $a = b > 1$, and $\NP$-hard otherwise. 
\end{reptheorem}

\begin{reptheorem}{4digraph}
    For all $b \geq a \geq 1$, $[a,b]\hHAM$ is $\NP$-hard in directed graphs of maximum degree 4. 
\end{reptheorem}

The hardness proofs for maximum degree digraphs easily extend to show hardness in regular digraphs.

\subsection{Preliminaries}
\paragraph{Introducing Depletors}
Now we introduce the notion of a \defn{depletor} gadget which will be used in all of our hardness reductions.
We introduce this notion by showing that $[a,b]\hHAM$ is hard in unrestricted graphs.
\begin{lemma}[Broersma, G{\"o}bel \cite{broersmakwalks}] \label{lem:basic}
    In unrestricted graphs, for any constants $1 \leq a \leq b$, $[a,b]\hHAM$ is \NP-hard. 
\end{lemma}
\begin{proof}
    We reduce from $[1,1]\hHAM$ (i.e. the standard Hamiltonian cycle problem). Given a graph $G$, we build a new graph $G'$ by gluing to each vertex $b - 1$ triangles, as shown in \cref{fig:glue-some-triangles}. In any $[a,b]\hHAM$ cycle each of these triangles must be visited --- this accounts for $b-1$ of the visits for each of the vertices in $G$, so the rest of the walk when restricted to $G$ must be a Hamiltonian cycle. For the other direction, note that it is possible to satisfy the visit requirements of the each triangle by stepping up to it, walking back and forth $a - 1$ times, and then stepping back down, contributing only $1$ additional visit to the vertex we attached the triangles to. So, if $G$ is Hamiltonian then $G'$ is $[a,b]\hHAM$.
\end{proof}

\begin{figure}[htb]
    \centering
    \includegraphics[width=0.4\linewidth]{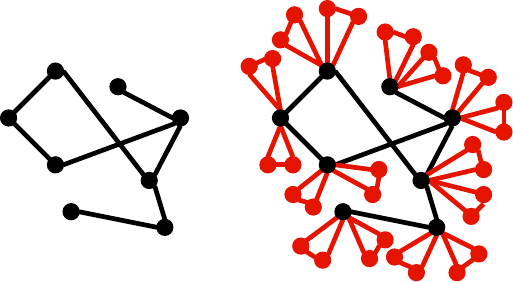}
    \caption{Reducing $[1,1]\hHAM$ to $[3,3]\hHAM$.}
    \label{fig:glue-some-triangles}
\end{figure}
In general a depletor gadget is some small graph that we glue to the starting graph $G$ that forces vertices to be visited many times. In \cref{lem:basic}, triangles are serving as depletors, since they each contribute a visit to the vertex they're glued to. In future sections, we will need depletors with stronger properties, such as forcing an attachment edge to be traversed many times.

\paragraph{Introducing Labellings}

Given an $[a,b]\hHAM$ cycle, we can label each edge of the graph $G$ with the number of times it's traversed in the cycle. These labels satisfy the following properties:

\begin{itemize}
    \item the sum of the labels on the edges adjacent to each vertex is an even number at least $2a$ and at most $2b$, and
    \item the subgraph on edges with non-zero labels is connected.
\end{itemize}

Now, given such a labeling, we could create a multigraph by including each edge with multiplicity given by its label. Since all vertices in this multigraph have even degree, we can (efficiently) find an Eulerian cycle in the multigraph, which is then an $[a,b]\hHAM$ cycle in the original graph. So, finding a labeling with these properties is equivalent to finding an $[a,b]\hHAM$ cycle; this perspective will make the problem easier to reason about.

\paragraph{A Technical Note}
Finally, we remark that, throughout this paper, degree bounds will be assumed to
be greater than or equal to $3$. All problems we consider are trivial in graphs
of maximum degree $2$, and so we will refrain from mentioning this special case.
We will use ``maximum degree'' to refer to any upper bound on the degrees of a
graph. In particular, a graph with no degree larger than $d$ is a max-degree $d$
graph even if it does not have a vertex achieving that degree. In directed
graphs, degree refers to the sum of in-degree and out-degree.

Another piece of notation that is important to note is that we will write $x \mod y$ to denote the remainder when $x$ is divided by $y$ (as opposed to the equivalence class $x+y\Z$).

\section{Complexity characterization in $d$-regular graphs}

\subsection{Easiness}
We claim that $[a,b]\hHAM$ can be solved in polynomial time for $d$-regular graphs, as long as either $d$ is even and $b \geq d/2$, or $d$ is odd and $b \geq d$. The reason for this is simple: the following two lemmas demonstrate that \emph{any} connected $d$-regular graph is $[a,b]\hHAM$ for such $b$.

\begin{lem}\label{lem:eulerit}
    For even $d$, any $d$-regular graph is $[d/2,d/2]\hHAM$. For odd $d$, any $d$-regular graph is $[d,d]\hHAM$.
\end{lem}
\begin{proof}
    If $d$ is even, then the graph is Eulerian --- taking an Euler tour will visit every vertex exactly $d/2$ times. If $d$ is odd, we can duplicate all its edges to get a $2d$-regular multigraph --- an Euler tour in that multigraph visits every vertex exactly $d$ times.
\end{proof}

Now we use Tutte's two-factor theorem to extend \cref{lem:eulerit} to more $[a,b]$. Specifically, the form of Tutte's two-factor theorem that we use is:
\begin{fact}\label{thm:two-factor}
   If $H$ is a $2k$-regular multigraph, its edges can be partitioned into $k$ distinct $2$-factors (i.e. $2$-regular sub-multigraphs).
\end{fact}
Using \cref{thm:two-factor} we obtain:

\begin{lem}[Broersma, G{\"o}bel \cite{broersmakwalks}]
    If a graph $G$ is $[b,b]\hHAM$ for some $b$, then it's also $[b', b']\hHAM$ for all $b' > b$.
\end{lem}
\begin{proof}
Let $H$ be the multigraph corresponding to a $[b,b]\hHAM$ walk on $G$. This is a $2b$-regular multigraph, so we can find a multigraph two-factor by Tutte's two-factor theorem. Duplicating all edges in that two-factor, we have a multigraph corresponding to a $[b+1, b+1]\hHAM$ walk on $G$. Repeating this process gives the desired claim. 
\end{proof}

Since a $[b,b]\hHAM$ cycle is also an $[a, b]\hHAM$ cycle for $a < b$, we've shown the easiness part of the claim.

\subsection{Hardness}
\subsubsection{Depletors}
The first step in the hardness argument is to construct depletors. We would like a small structure that can always be part of an $[a,b]\hHAM$ cycle as long as it's visited at least once. 

We will describe such a $d$-regular depletor for even $d$ with two attachment edges and $2d$ vertices. Add an edge between $v_i$ and $v_j$ whenever $(i - j) \mod 2d \in \{-d/2, \dots, d/2\}$. Now, cut the edges $(v_{2d}, v_{3d/2})$ and $(v_1, v_{d/2+1})$, and add the edge $(v_{3d/2}, v_{d/2+1})$ --- note that this is in fact a new edge since 
$$(3d/2-d/2-1)\bmod 2d\not\in \{-d/2, \dots, d/2\}.$$ So, if we add an edge leaving the gadget from each of $v_{2d}$ and $v_1$, the construction will be $d$-regular.

\begin{figure}[htb]
    \centering
    \includegraphics[width=0.5\linewidth]{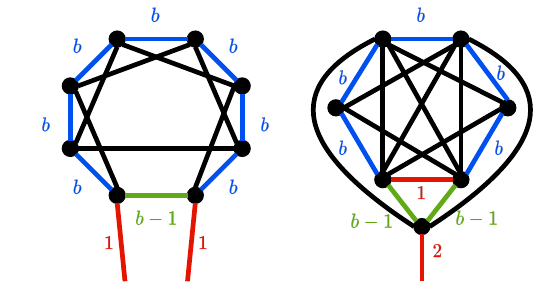}
    \caption{Labellings corresponding to $[b,b]$ walks on 4-regular and 5-regular depletors, respectively.}
    \label{fig:regular-depletors}
\end{figure}

\begin{lemma}\label{lem:evendepletor}
    For any $b \geq a \geq 1$, there is a valid $[a,b]\hHAM$ traversal of this even-regularity depletor in which each attachment edge is traversed exactly once.
\end{lemma}
\begin{proof}
    Assign each edge $(v_i, v_{i+1})$ of the gadget to be traversed $b$ times, and assign the edge $(v_{2d},v_1)$ to be traversed $b-1$ times; see \cref{fig:regular-depletors}. 
\end{proof}

Note that this gadget required two attachment edges. This is a necessary consequence of the regularity constraint: since $d$ is even, there must be an even number of half-edges within the gadget, meaning that the number of edges leaving the gadget must be even. 

When $d$ is odd, we can construct a $d$-regular depletor needing only one
leaving edge. Start with the complete graph $K_{d + 1}$, and then delete all
edges $(v_{i}, v_{i + (d+1)/2})$ except for the edge $(v_1, v_{(d+1)/2 + 1})$.
At this point, every vertex has degree $(d-1)$, except for $v_1$ and $v_{(d+1)/2
+ 1}$, which have degree $d$. So, we add a $(d+2)$-th vertex $v^*$, with edges
to every other vertex except $v_1$ and $v_{(d+1)/2 + 1}$, and a single edge
leaving the construction --- note that the result is now $d$-regular.

\begin{lemma}\label{lem:odddepletor}
    For any $b \geq 2$, there is a valid $[a,b]\hHAM$ traversal of this odd-regularity depletor in which the attachment edge is traversed exactly twice.
\end{lemma}

\begin{proof}
    Assign all the edges in the loop passing through all vertices to be traversed $b-1$ times, and then add one more traversal to the edges $(v_i, v_{(i + 1) \bmod d + 2})$; see \cref{fig:regular-depletors}.
\end{proof}

\subsubsection{Reductions}

These depletor gadgets let us use the construction from \cref{lem:basic} to resolve the case of even $d$.

\begin{lemma}
    For even $d$, if $a \leq b < d/2$, then $[a,b]\hHAM$ is $\NP$-hard in $d$-regular graphs.
\end{lemma}
\begin{proof}
    We reduce from $[1,1]\hHAM$ in $(d - 2b + 2)$-regular graphs, which is known to be $\NP$-hard as long as $(d - 2b + 2) > 2$ --- i.e. as long as $b < d/2$. Given a $(d - 2b + 2)$-regular graph $G$, we construct $G'$ by attaching $(b-1)$ copies of the $d$-regular depletor of \cref{lem:evendepletor} to each vertex. Recall that this depletor required two attachment edges, so $G'$ is $d$-regular. But now, note that in the process of visiting each depletor gadget once, each vertex from $G$ must be visited at least $b-1$ times, leaving room for only $1$ additional visit. So, if $G$ was not Hamiltonian, $G'$ is not $[a,b]$-Hamiltonian for any $a$. Conversely, if $G$ was Hamiltonian, then by \cref{lem:evendepletor} we can extend the Hamiltonian cycle to a $[b,b]\hHAM$ cycle by visiting each depletor exactly once. 
\end{proof}

This argument is not immediately sufficient for the case of odd regularity, however. Applying the same construction would yield the following:

\begin{lemma}
    For odd $d$, if $a \leq b < d - 1$, then $[a,b]\hHAM$ is \NP-hard in $d$-regular graphs.
\end{lemma}

To be tight with our easiness results, however, we would like to show hardness of $[a,b]\hHAM$ when $b < d$. The gap here is coming from our base case: the hard problem we have to reduce from in $3$-regular graphs is $[1,1]\hHAM$, whereas we need to show that even $[1,2]\hHAM$ and $[2,2]\hHAM$ are hard. Fortunately, a slightly more involved reduction closes this gap.

\begin{lemma}\label{lem:oddreghard}
    For odd $d$, if $a \leq d-1$, then $[a,d-1]\hHAM$ is \NP-hard in $d$-regular graphs.
\end{lemma}
\begin{proof}
    We reduce from $[1,1]\hHAM$ in $3$-regular graphs. Let $G$ be a $3$-regular graph; we construct $G'$ as follows:
    \begin{itemize}
        \item Glue $d - 3$ depletors to each vertex in $G$.
        \item Trisect each of $G$'s original edges --- that is, replace them with paths containing $2$ intermediate nodes.
        \item Glue $d - 2$ depletors to each of the intermediate nodes just added.
    \end{itemize}
    Using the $d$-regular depletors of \cref{lem:odddepletor}, this construction is $d$-regular. We would like to show that $G'$ is $[d-1,d-1]$ Hamiltonian if $G$ is Hamiltonian, and otherwise not even $[1,d-1]$ Hamiltonian --- this suffices to prove the lemma. 

    \begin{figure}[htb]
    \centering
    \includegraphics[width=0.45\linewidth]{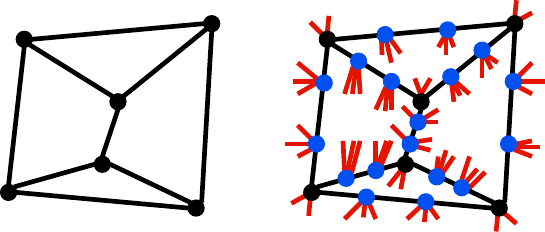}
    \caption{Reduction from $[1,1]\hHAM$ in $3$-regular graphs to $[4,4]\hHAM$ in $5$-regular graphs. Red outgoing edges represent depletors.}
    \label{fig:3-reduction}
\end{figure}

    First, note that if you ignore the depletors, each intermediate node must be visited exactly once, since the depletors contribute $d-2$ of the $d-1$ allowed visits. Thus, each edge gadget must be labeled either $1, 1, 1$ or $2, 0, 2$; see \cref{fig:edge-gadget}. 
    
    \begin{figure}[htb]
    \centering
    \includegraphics[width=0.20\linewidth]{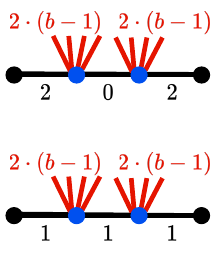}
    \caption{The two possible labelings of the edge gadget.}
    \label{fig:edge-gadget}
\end{figure}
    
    The depletors contribute $d-3$ visits to each of the original vertices, so each has an additional label budget of only $4$; this means that for every original vertex, exactly two of the three incident edge gadgets must be labeled $1,1,1$. So, the edges of $G$ corresponding to $1,1,1$ edge gadgets in $G'$ form a 2-factor; see \cref{fig:3-reduction-demonstration}. But now, suppose this $2$-factor had multiple connected components. Since $2,0,2$ edges do not create connectivity, this would mean that the labeling of $G'$ was not connected, and so couldn't correspond to a $[1, d-1]\hHAM$ walk. Thus, $G'$ can be $[1,d-1]\hHAM$ only if $G$ had a connected 2-factor, i.e. a Hamiltonian cycle.\\

    \begin{figure}[htb]
    \centering
    \includegraphics[width=0.3\linewidth]{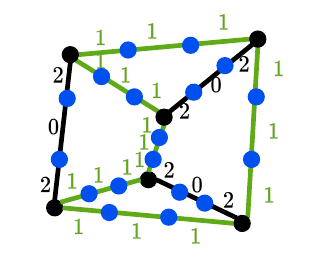}
    \caption{The edge gadgets labeled $1,1,1$ must correspond to a Hamiltonian path.}
    \label{fig:3-reduction-demonstration}
\end{figure}

    The reverse direction follows by the same logic: given a Hamiltonian cycle in $G$, we label all edge gadgets corresponding to the cycle $1,1,1$ and all edge gadgets not in the cycle $2,0,2$; visiting each depletor once gives a $[d-1, d-1]\hHAM$ labeling.

\end{proof}

Combining the above lemmas gives a complete characterization of the hardness of $[a,b]\hHAM$ in regular graphs.

\begin{theorem}\label{regular}
    When restricted to $d$-regular graphs, $[a,b]\hHAM$ is \NP-hard if either $d$ is even and $b < d/2$, or $d$ is odd and $b < d$. Otherwise, $[a,b]\hHAM \in \P$.
\end{theorem}

\section{Complexity characterization in graphs of max-degree $d$}

We now have a tight understanding of when $[a,b]\hHAM$ is hard in regular graphs. If, instead, we just cared about graphs with constant maximum degree, the hardness results would translate, but easiness would not. We show in this section that in fact the true threshold for hardness in this setting looks very different: complexity is determined not by $b$ but by the ratio $\frac{b}{a}$. 

For our easiness result, we will note that testing $[a,b]$-Hamiltonicity is easy in trees, and then show that when $\frac{b}{a} \leq d-1$, for $a > 1$, any failure of $[a,b]$-Hamiltonicity is witnessed by a tree-like subgraph. To show hardness, we will exploit a few low-degree vertices to construct and fine-tune highly efficient depletor gadgets. 

\subsection{Easiness}\label{sec:maxdegez}

First, we observe that without loss of generality an $[a,b]\hHAM$ walk never traverses a given edge more than $2a$ times:
\begin{lemma}\label{lem:maxlabel}
Let $G$ be a graph. If $G$ has an $[a,b]\hHAM$ labeling, it has such a labeling with maximum label at most $2a$.
\end{lemma}
\begin{proof}
    Take such an $[a,b]\hHAM$ walk, and consider an edge with label more than $2a$. Each of the vertices incident to that edge is therefore visited at least $a+1$ times by the walk --- so, if we subtract $2$ from the label of that edge, each of them will still be visited at least $a$ times, and the labeling will still be valid.
\end{proof}

Now, we give an algorithm for $[a,b]\hHAM$ in trees, which will be an important component of our algorithm for $[a,b]\hHAM$ in general graphs.

\begin{lemma}\label{lem:treeham}
Let $b \geq a \geq 1$, and let $G$ be a tree. There is a linear-time algorithm to determine whether $G$ is $[a,b]\hHAM$. 
\end{lemma}
\begin{proof}
Arbitrarily root the tree. We will work from the bottom-up, recording for each edge the set of label values that can extend to a valid labeling in the subtree below it. 
In fact, we will show that for each edge $e$ the set of labels for $e$ that can extend to a valid labelling in the subtree below $e$ is of the form $\set{\ell_e, \ell_e+2,\ldots, u_e}$ for even integers $\ell_e, u_e$. Thus, it will suffice to record for each edge $e$ the minimum and maximum labels $\ell_e, u_e$ which can extend to valid labellings in the subtree below $e$.

Start by setting $\ell_e=2a, u_e = 2b$ for each edge $e$ incident to a leaf of the tree: these edges must be traversed at least $2a$ times in order to visit the leaves $a$ times.
Now we repeatedly apply the following process to find $u_e,\ell_e$ for all edges in the tree:

Find a vertex $v$ such that all edges in the subtree rooted at $v$ have been labelled.
Let $L_v,U_v$ be the sum of $\ell_{e'},u_{e'}$ respectively over edges $e'$ going down from $v$.
Let $e$ be the edge going up the tree from $v$.
We claim that the valid labels for $e$ are precisely the even numbers in the interval $[\ell_e, u_e]$ defined by 
$$\ell_e = \max(2, 2a-U_v),\quad u_e = 2b-L_v.$$
First we argue that labelling for $e$ must lie in this set; then we will argue that any labels in this set can be extending to labellings on the subtree.
The label on $e$ must be even because all of the edges going down from $v$ have even labels (by induction). The label on $e$ cannot be zero, or else the labelling is not connected.
The label on $e$ must be at least $2a-U_v$ or else vertex $v$ is not visited enough times.
The label on $e$ cannot be larger than $2b-L_v$ or else vertex $v$ is visited too many times.
It is clear that for any even $k\in [\ell_e, u_e]$, if we give  $e$ label $k$, then there is a compatible labelling of the subtree rooted at $v$. 

If at any step in the process an edge $e$ is given $\ell_e > u_e$, then we deduce that there is no valid labelling of the tree.
Otherwise, the process will terminate with every edge having $\ell_e, u_e$ with $\ell_e\le u_e$.
To finish our check of whether the tree is $[a,b]\hHAM$ we simply iterate over all possible ways of assigning labels to the edges incident to the root, taken from their corresponding allowed sets $[\ell_e, u_e]\cap 2\Z$. 
If there is any combination of labels of the edges incident to the root that satisfies the root, then by construction of the allowed label sets, we can extend the labellings given to each $e$ incident to the root to the entire subtree below $e$. Thus, we can label the entire tree.
\end{proof}

We now use this procedure for trees as part of an algorithm for $[a,b]\hHAM$ in general graphs in the case $a>1$. Afterwards we will analyze the case $a = 1$, which is simple but (somewhat surprisingly) meaningfully different from the case $a>1$.

\begin{lemma}\label{lem:a>1oof}
For any $d$ and any $b \geq a > 1$ such that $\frac{b}{a} \geq (d-1)$, there exists a polynomial-time algorithm to test whether a graph $G$ of max degree $d$ is $[a,b]\hHAM$.
\end{lemma}
\begin{proof}

The first step of the algorithm is to search for \defn{tree derivations}. Call an edge $e \in E(G)$ \defn{dangerous} if it does not belong to any cycle in $G$ (the set of dangerous edges can be compute easily using connectivity tests on modified versions of the graph). The tree derivation stage will assign labels to all dangerous edges. 

The means of performing a derivation is exactly as described in \cref{lem:treeham}. The set of dangerous edges forms a forest, so for each tree in the forest we run the algorithm from \cref{lem:treeham} to find an $[a,b]\hHAM$ labeling that satisfies all vertices except possibly the root. If any of these tree derivations fail, $G$ cannot be $[a,b]\hHAM$ --- the failure of a tree derivation gives a set of vertices that cannot be mutually satisfied under any labeling, since there can't be any more edges into this set.

\begin{figure}[htb]
    \centering
    \includegraphics[width=.55\linewidth]{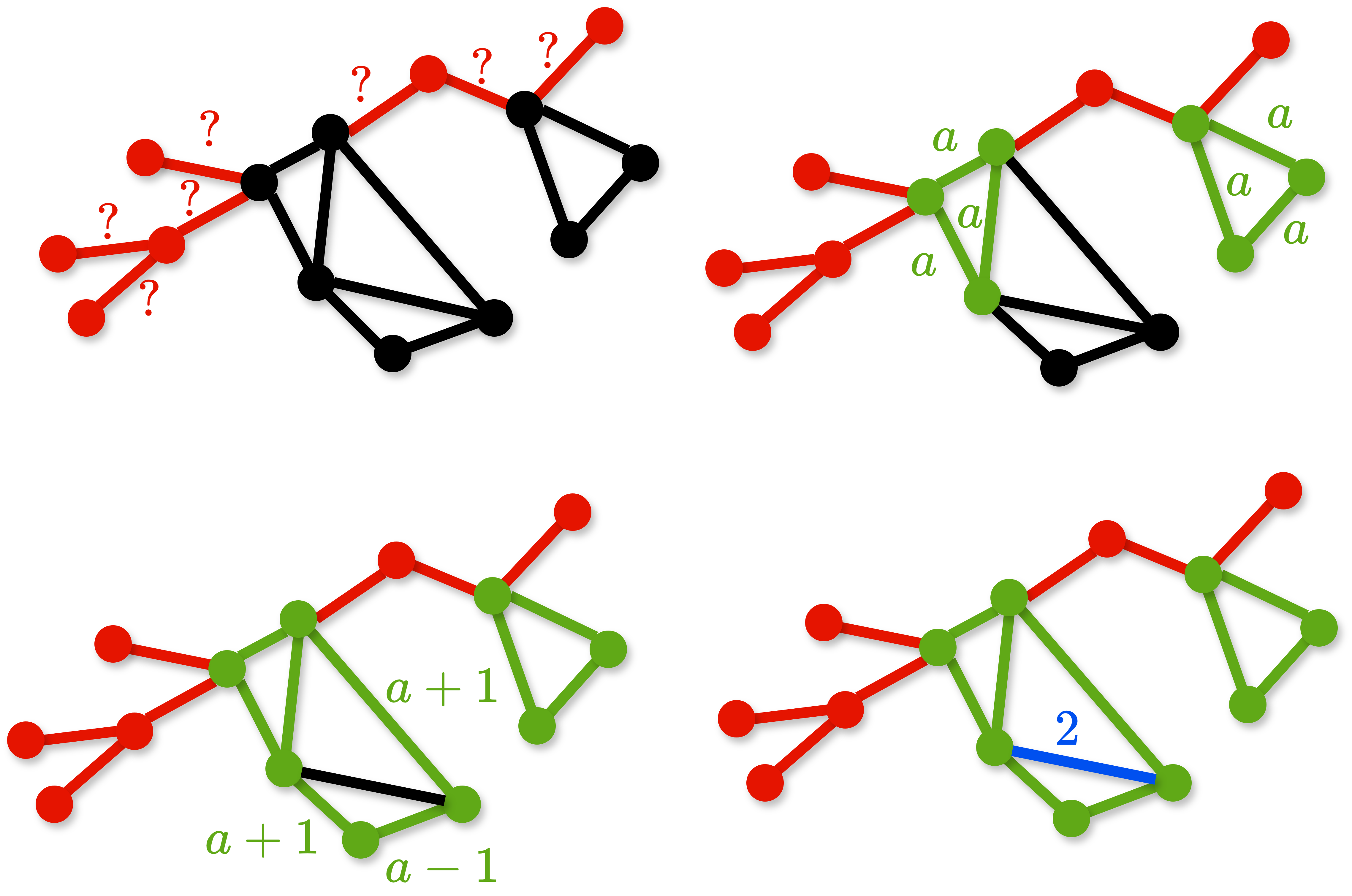}
    \caption{The algorithm performs all tree derivations, then labels disjoint cycles, then merges paths into those cycles, and finally labels all remaining edges with $2$.}
    \label{fig:enter-label}
\end{figure}

If none of these tree derivations fail, we have a labeling of all dangerous edges that satisfies all vertices not involved in any cycle. Call a vertex \defn{good} if it's involved in at least one cycle. The next step of our argument is to show that, no matter what labeling we choose for the dangerous edges, there is always an assignment to the remaining edges that visits each good vertex between $a$ and $(d-1)a \leq b$ times, showing that $G$ is necessarily $[a,b]\hHAM$ unless a tree derivation fails. 

\paragraph{Existence of a labelling, given that tree derivations succeed}
We now give a procedure which finds a labelling of the remainder of the graph once all tree derivations have succeeded; we stress that we don't actually run this procedure in our algorithm, it is an analysis tool to demonstrate the existence of a labelling.
The procedure is as follows:

\begin{itemize}
    \item Find an arbitrary cycle in the graph, and label all of its edges with $a$. Mark the vertices involved in that cycle \defn{happy}.
    \item While there remains a cycle involving only unhappy vertices, repeat the above step with that cycle.
    \item Now, suppose there still exists an unhappy good vertex. Since good vertices are involved in cycles, there must exist a path that starts and ends in happy vertices but otherwise contains only unhappy good vertices. (Here, we're allowing that ``path'' to potentially start and end at the same happy vertex --- this is a slight abuse of terminology.) Find such a path. If $a$ is even, label all edges in the path with $a$. Otherwise, give the edges in the path alternating labels of $a+1$ and $a-1$. Mark all vertices along the path happy.
    \item Repeat the above step until all good vertices are happy. Then, label all unlabeled edges of the graph with $2$.
\end{itemize}

Now we analyze the procedure. First, note that all edges are assigned nonzero
weight, so this construction is connected\footnote{This is the place where the
argument fails for $a = 1$. In that case, the edges assigned $a-1$ would be
given zero weight, and so we wouldn't ensure connectivity.}. Labeling a full
cycle with $a$'s increases every vertex label sum by either $0$ or $2a$, and all
other labels used in this construction are even values, so this labeling gives
even sums into every vertex. So, the only point to argue is that every vertex is
given sum between $2a$ and $2b$. This holds for the non-good vertices by virtue
of the tree derivations succeeding. Now, note that for each good vertex, at some
point we label two of its incident edges either $a$ and $a$ or $a+1$ and $a-1$.
In either case, this ensures that they have total weight at least $2a$. By
\cref{lem:maxlabel}, we can decrease all labels assigned while performing tree
derivations to be at most $2a$. Thus, the sum of the labels incident to any
good vertex is at most $$2a+(d-2)2a  =(d-1)2a\le 2b.$$
\end{proof}

The final point to note in this section is the following, covering the $a=1$ case:
\begin{lemma}\label{lem:a=1lol}
    Every connected graph of maximum degree $d$ is $[1,d]\hHAM$
\end{lemma}
\begin{proof}
    Label all edges $2$. This visits each vertex at least once, and at most $d$ times.
\end{proof}

Together, \cref{lem:a>1oof} and \cref{lem:a=1lol} give our easiness results. Note that we have easiness for $[a,b]\hHAM$ whenever $\frac{b}{a} \geq d-1$, unless $a = 1$, in which case we need $\frac{b}{a} \geq d$. The first thing we will observe in the next section is that $[1, d-1]\hHAM$ is in fact $\NP$-hard, so this distinction is necessary.

\subsection{Hardness}
\begin{lemma}\label{lem:a=1hard}
    $[1,d-1]\hHAM$ is $\NP$-hard in graphs of maximum degree $d$.
\end{lemma}
\begin{proof}
    When $d$ is odd, this follows by our hardness results in \cref{lem:oddreghard} --- a $d$-regular graph is in particular a graph of maximum degree $d$. When $d$ is even, we can use the exact same construction as in the proof of that lemma, but just use 3-regular depletors instead of $d$-regular depletors, since we're no longer concerned with regularity.
\end{proof}

When $a > 1$, though, we will show hardness for much larger values of $b$ than
are hard in $d$-regular graphs. Proving this claim will require a new depletor
construction, which we call a \defn{chain depletor}, since it's composed of a
long chain of subunits:

\begin{lemma}\label{lem:chains}
Let $a,b$ be integers with $1\le a\le b < a(d-1)$.
For any integer $2k\in [2,2a]$, there exists a gadget of maximum degree $d$ and with a single attachment edge, which can be part of an $[a,b]\hHAM$ cycle if the attachment edge is labeled $2k$, but cannot be satisfied if the attachment edge is labeled $< 2k$.
\end{lemma}
\begin{proof}
    Consider the gadget shown in \cref{fig:chainlink} --- a $3$-edge path on two intermediate vertices, with depletors of total depletion $D$ attached to the second intermediate vertex. 
    
    \begin{figure}[htb]
        \centering
        \includegraphics[width=0.22\linewidth]{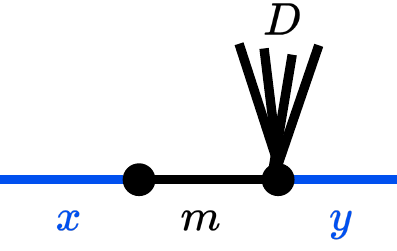}
        \caption{A ``chainlink'' gadget, forcing $y \leq x + (2b - 2a - D)$}
        \label{fig:chainlink}
    \end{figure}

    Gadgets of this form will serve as the links in the chain. Let $x$ be the label leaving the first intermediate vertex, $y$ be the label leaving the second intermediate vertex, and $m$ be the label on the edge between the two vertices. In order to satisfy the visit counts of the first vertex, we must have $x + m \geq 2a$, and in order to satisfy the visit counts of the second vertex we must have $y + m + D \leq 2b$. Adding these constraints together, we find that if $y > x + (2b - 2a - D)$, there is no setting of $m$ that satisfies both of them. On the other hand, if $x \in [2, 2a - 2]$, then setting $m = 2a - x$ and $y = x + (2b - 2a - D)$ not only satisfies all visit count constraints, but also preserves connectivity \emph{so long as $y\ge 2$}. So, this gadget forces its ``outgoing'' edge to have label at least $(2b - 2a - D)$ more than its ``incoming'' edge, and is satisfied by an outgoing label exactly that large as long as its incoming edge had label between $2$ and $2a-2$. 
    
    Of course, in constructing this link gadget, we assumed we were able to deplete a vertex by $D$. So if we want to use these gadgets for anything, we first need some sort of depletors to start with. Fortunately we do have some depletors: using for instance the $3$-regular depletor of \cref{lem:odddepletor} (or an edge with a triangle on the end) gives $2$ depletion, and gluing a single vertex along an edge gives $2a$ depletion. The overall goal of the chain depletor construction is to get a single gadget with depletion in between $2$ and $2a$, since for our reductions we will need to fine-tune the depletion exactly. 

    For our construction, we chain together multiple copies of the link gadget, identifying the outgoing edge of one with the incoming edge of the next, and finally ending in a single vertex, as shown in \cref{fig:chain}:
    
    \begin{figure}[htb]
        \centering
        \includegraphics[width=0.5\linewidth]{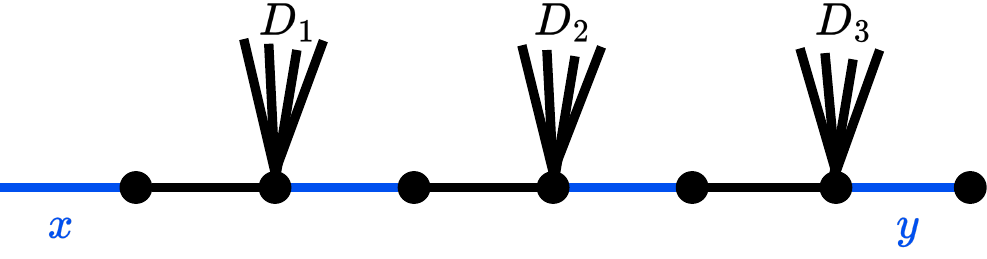}
        \caption{A chain of link gadgets forcing \ $x + \sum_{i=1}^\ell (2b - 2a - D_i) \geq y \geq 2a$.}
        \label{fig:chain}
    \end{figure}

    Suppose the link gadgets are constructed with depletions $D_1, \dots, D_\ell$, and let $x$ and $y$ be the labels assigned to the attachment edge and final edge of the chain, respectively. By our analysis of the link gadget, this construction is unsatisfiable if the final edge has label $y > x + \sum_{i=1}^\ell (2b - 2a - D_i)$. We would therefore like to choose the depletions $D_i$ such that $\sum_{i=1}^\ell (2b - 2a - D_i) = 2a - 2k$, since then if the incoming edge has too small a label $x < 2k$, this will make setting $y \geq 2a$ impossible, and so ensure that the vertex at the end of the chain can't be satisfied. We also want to make sure that the chain \emph{is} satisfiable if the incoming edge is labeled $x = 2k$; to guarantee that, it suffices by our analysis of the link gadgets to ensure that $2k + \sum_{i=1}^j (2b - 2a - D_i) \in [2, 2a - 2]$ for all $j < \ell$, since then the construction is satisfied by greedily labeling each edge, from attachment edge to the final edge, with the smallest allowable value. 

    We will now show that, except in the case where $b$ is a multiple of $a$, which we will handle slightly differently, it is always possible to construct such a sequence of $D_i$. Let $M = \left\lfloor \frac{b}{a} - 1\right\rfloor$ and $L = 2b \mod 2a$, so that $2b = 2a(M + 1) + L$. We will use three kinds of links, with the following depletions:
    \newcommand{\larg}{\text{large}}
    \newcommand{\medi}{\text{medium}}
    \newcommand{\smal}{\text{small}}
    \begin{itemize}
        \item $D_{\larg} = 2a(M+1)$, created by gluing $(M+1)$ single edges. We have 
        \[(2b - 2a - D_{\larg}) = L - 2a.\]
        \item $D_{\medi} = 2aM + 2$, created by gluing $M$ single edges and one 3-regular depletor. We have 
        \[(2b - 2a - D_{\medi}) = L - 2.\]
        \item  $D_{\smal} = 2aM$, created by gluing $M$ single edges. We have 
        \[(2b - 2a - D_{\smal}) = L.\]
    \end{itemize}
    To verify that these links are well defined, we must check that the maximum degree in the link is at most $d$.
   Note that $\frac{b}{a} < d-1$ by assumption.
   Thus, $M \leq d-3$, so the degree constraint allows us to glue $M+1$ depletors to a vertex in our link gadget.

   As long as $b$ is not a multiple of $a$, we have that $L$ is positive, $L -
   2a$ is negative, and $\gcd(L, L-2, L-2a) = 2$. By Bezout's lemma, there must
   therefore exist non-negative integers $w_1, w_2, w_3$ such that $$w_1(L - 2a)
   + w_2(L-2) + w_3L = 2a - 2k.$$ So, we have constructed a multiset of links with
   the desired sum. Now we give an ordering of the links in the multiset such that
   such that $$2k + \sum_{i=1}^j (2b - 2a - D_i) \in [2, 2a - 2]$$ for all $j <
   \ell$. We construct the ordering greedily. Let $S$ be the
   multiset of links we found above. While $$2k + \sum_{i=1}^j (2b - 2a - D_i)
   \neq 2a,$$ we choose an element $D^*$ of $S$ such that $$(2b - 2a - D^*) +
   2k + \sum_{i=1}^j (2b - 2a - D_i) \in [2, 2a],$$ set $D_{j+1} = D^*$, and
   remove $D^*$ from $S$. 
   We now show that such a $D^*$ necessarily exists. 
   
   Suppose $S$ is nonempty, and adding a large link brings the sum outside the desired interval. That is, 
   \[(L-2a) + 2k + \sum_{i=1}^j (2b - 2a - D_i) \leq 0.\]
   Then, since adding all remaining elements of $S$ would bring the sum to $2a$, we know that $S$ cannot consist only of large links. $L$ and $L-2$ are both non-negative, and we know
   \[L + 2k + \sum_{i=1}^j (2b - 2a - D_i) \leq 0 + 2a = 2a,\]
   so adding either a medium or a small link would be valid. If, on the other hand, adding a large link would be valid, then either $S$ contains a large link and we can add that, or $S$ contains only small and medium links, in which case since $L$ and $L-2$ are both non-negative we can add them in any order. Repeating this argument on $S$ until we reach a sum of $2a$ will give a greedy construction for the link sequence $D^*$.\\

    We have established the lemma in the case that $a \nmid b$; we now deal with the case where $b$ is a multiple of $a$. Here, $L = 0$, so now none of $L$, $L-2$, and $L-2a$ are positive. In this case, we will build our chain using $k-1$ copies of the $D_{\medi}$ link, and feed the final edge into a $3$-regular depletor as opposed to a single vertex, as in \cref{fig:chain2}. 
    \begin{figure}[htb]
        \centering
        \includegraphics[width=0.55\linewidth]{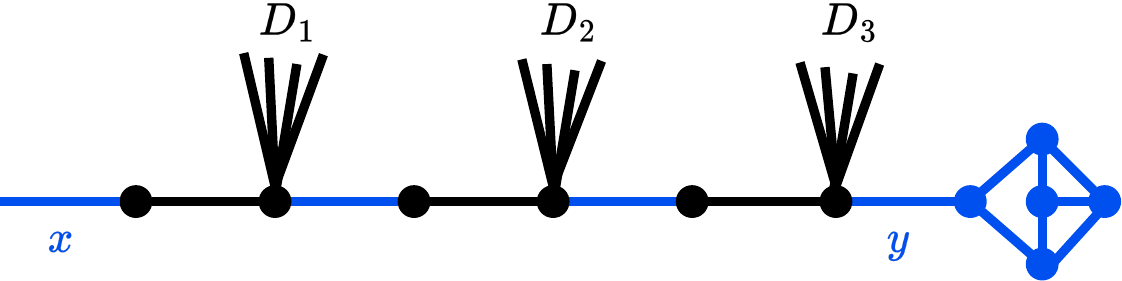}
        \caption{A chain of link gadgets forcing\ $x - 2(k-1) \geq y \geq 2$.}
        \label{fig:chain2}
    \end{figure}
    
    If the attachment edge is labeled $2k$, then the final edge can be labeled at least $2$ and so this is satisfiable. However if the attachment edge is labeled $x < 2k$, then the final edge can't have positive label, and so the final depletor is not visited.
\end{proof}

Now that we have these fine-tunable depletors, we can glue them to the vertices of a graph to contribute a specific number of artificial visits. However, in order to contribute \emph{enough} visits in our reduction, it is sometimes necessary to also deplete along the edges. 

\begin{lemma}\label{lem:maxdeg-edge-gadget}
Let $a,b$ be constants with $1\le a\le b < a(d-1)$.
For any integer $2k\in [2, 2a-2]$, there exists a gadget of maximum degree $d$ and with two attachment edges with the following properties:
\begin{itemize}
\item The gadget can be satisfied if both attachment edges are labeled $2k$.
\item The gadget cannot be satisfied if an attachment edge is labeled less than $2k$.
\item  The gadget can be both satisfied and connected if both attachment edges are labeled $2k + 1$.
\item The gadget cannot be connected if an attachment edge is labeled less than $2k + 1$.
\end{itemize}
\end{lemma}
\begin{proof}
Consider the gadget shown in \cref{fig:maxdeg-edge-gadget}: a path on $4$ intermediate vertices, where the $2$ middle vertices each have $2b - 2a + 2k$ worth of depletion glued to them --- this can be constructed using the chain depletors of \cref{lem:chains}.
    \begin{figure}[htb]
        \centering
        \includegraphics[width=0.4\linewidth]{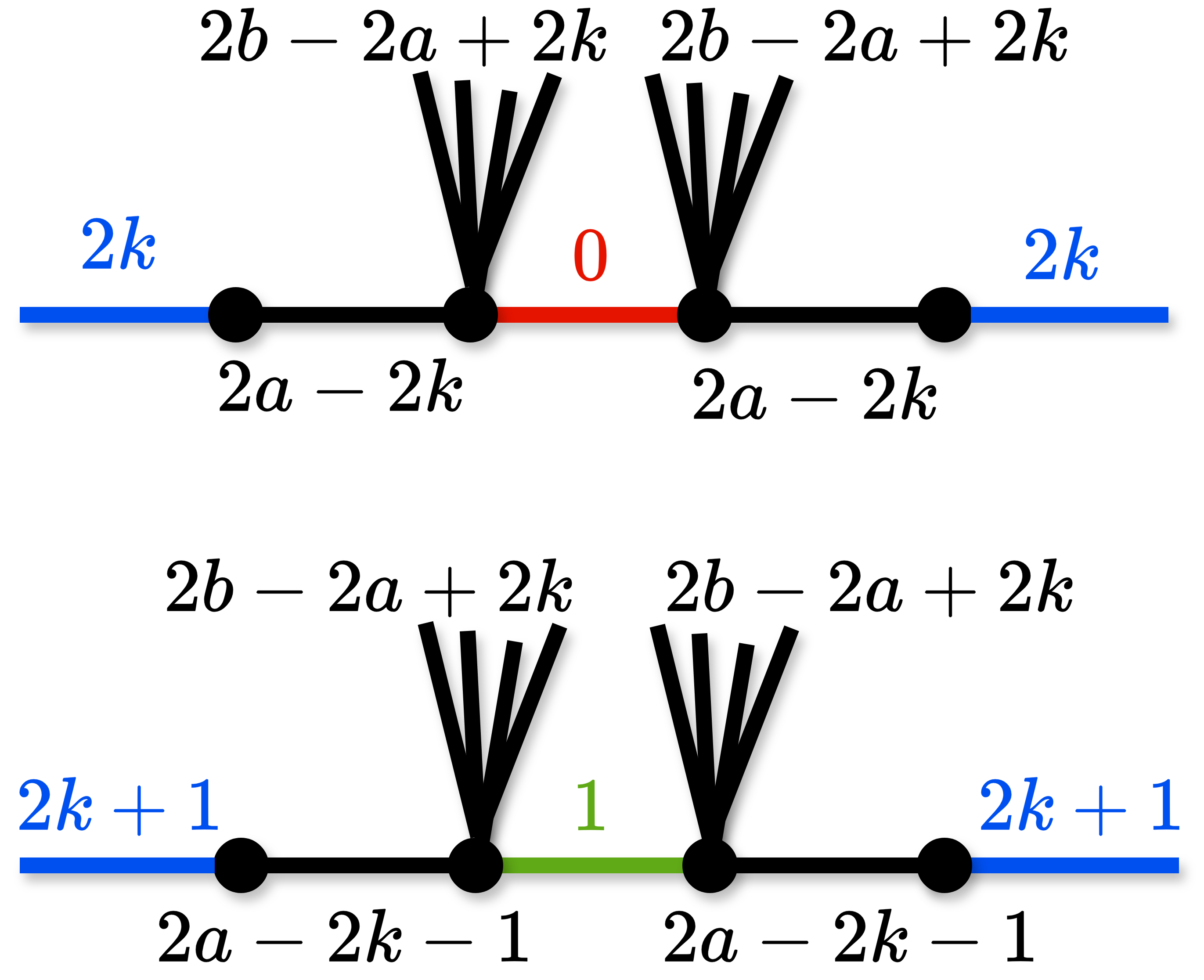}
        \caption{An edge gadget that is satisfiable if both inputs are $2k$, satisfiable and connected if both inputs are $2k+1$.
        } 
        \label{fig:maxdeg-edge-gadget}
    \end{figure}
     The middle edge can be assigned non-negative label only if both attachment edges are labeled at least $2k$, and can be assigned positive label only if both attachment edges are labeled at least $2k + 1$.
\end{proof}

We can think of this edge gadget as effectively providing $2k$ depletion to both sides. So, we can now give a reduction to show our complete complexity characterization:

\begin{theorem}\label{maxdeg}
    When restricted to graphs of maximum degree $d$, $[a,b]\hHAM$ is \NP-hard if either 
    $a = 1$ and $b < d$, or
    $a>1$ and $\frac{b}{a} < d-1$.
    Otherwise, $[a,b]\hHAM \in \P$.
\end{theorem}
\begin{proof}
    The easiness results were established in \cref{sec:maxdegez}, and the case
    of $a = 1$ in \cref{lem:a=1hard}. So, it suffices to show a reduction from
    $[1,1]\hHAM$ in a $3$-regular graph $G$ to $[a,b]\hHAM$ in a max-degree $d$
    graph $G'$ when $1 < a \leq b < (d-1) a$. By replacing the edges of $G$ with
    the edge gadgets of \cref{lem:maxdeg-edge-gadget}, we can introduce up to
    $3(2a - 1)$ depletion at each vertex, and by gluing up to $(d-3)$ additional
    chain depletors to each original vertex, we can introduce up to $2a(d-3)$
    more. Since $$3(2a - 1) + 2a(d-3) \geq 2b - 2,$$ and all of these depletors are
    fully fine-tunable, we can ensure exactly $2b-2$ depletion on each original
    vertex of $G$. This means that up to but no more than $2$ of the edge
    gadgets incident to each original vertex can be connected; thus, $G'$ is
    $[a,b]\hHAM$ if and only if $G$ has a connected $2$-factor.
\end{proof}

\section{Complexity characterization in directed graphs}

In contrast to the case of undirected graphs, where we found that sufficiently large size of the $[a,b]$ range compared to degree made $[a,b]\hHAM$ easy, we will show that in directed graphs every case is hard, except when the $[a,b]$ range is very small (contains only a single element) and $d = 3$. 

Just as in the directed case, we can think of the problem of $[a,a]\hHAM$ as finding labels. More precisely, a directed graph is $[a,a]\hHAM$ if and only if there is a connected labeling of the edges such that every vertex has in-label  sum $a$ and out-label sum $a$. As in the directed case the equivalence of these notions follows from taking an Euler tour on the multigraph generated by viewing the labels as multiplicities for edges.

\subsection{Easiness}
First we give a simple algorithm for the easy cases of $[a,b]\hHAM$ in directed graphs.
\begin{lemma}\label{lem:3direz}
    For $a>1$, there exists a polynomial-time algorithm to decide $[a,a]\hHAM$ in directed graphs of maximum degree $3$. 
\end{lemma}
\begin{proof}
Fix a directed graph of maximum degree $3$.
If any vertex has in-degree or out-degree $3$ we immediately deduce that the
graph is not $[a,a]\hHAM$; in the remainder of the proof we assume there are
no vertices with in-degree or out-degree $3$.
    
Our algorithm for $[a,a]\hHAM$ begins as follows: while there exists a vertex $v$ in the graph with only a single unlabeled out-edge, or only a single unlabeled in-edge, assign that edge the unique label that satisfies $v$. When that process terminates, each vertex must have either zero or two unlabeled in-edges, and either zero or two unlabeled out-edges. So, if we look at the subgraph on unlabeled edges, and think of them as undirected, we must have a union of even-length cycles. For each of those cycles, assign labels by alternating between $\left\lfloor\frac{a}{2}\right\rfloor$ and $\left\lceil\frac{a}{2}\right\rceil$. Now, check whether the assigned labels represent an $[a,a]\hHAM$ labeling; if so accept, if not reject.
    
    Since $\left\lfloor\frac{a}{2}\right\rfloor + \left\lceil\frac{a}{2}\right\rceil = a$, if our labeling violates a vertex visit constraint, that constraint must have been violated during the forced initial phase. Similarly, since this process only labels edges with $0$ during the initial phase, if connectivity fails then no labeling can achieve connectivity.
\end{proof}

\subsection{Hardness}
\subsubsection{Max degree $3$}
We now show hardness of $[a,a+1]\hHAM$.
We will argue by induction. The following lemma constitutes the inductive step of the argument.
\begin{lemma}\label{lem:3regbasecase}
Fix $k\ge 1$. Suppose that $[1,k]\hHAM$ is \NP-hard in $3$-regular digraphs. Then, for any $a\ge 1$, $[a,a+k]\hHAM$ is \NP-hard in $3$-regular digraphs.
\end{lemma}
\begin{proof}
Let $G$ be a $3$-regular digraph. We construct $G'$ by replacing each vertex of
$G$ with a \defn{vertex gadget} consisting of three \defn{hub vertices}
connected by \defn{bridge gadgets}, as shown in \cref{fig:3-di-hard-level1}.
    \begin{figure}[htb]
        \centering
        \includegraphics[width=0.6\linewidth]{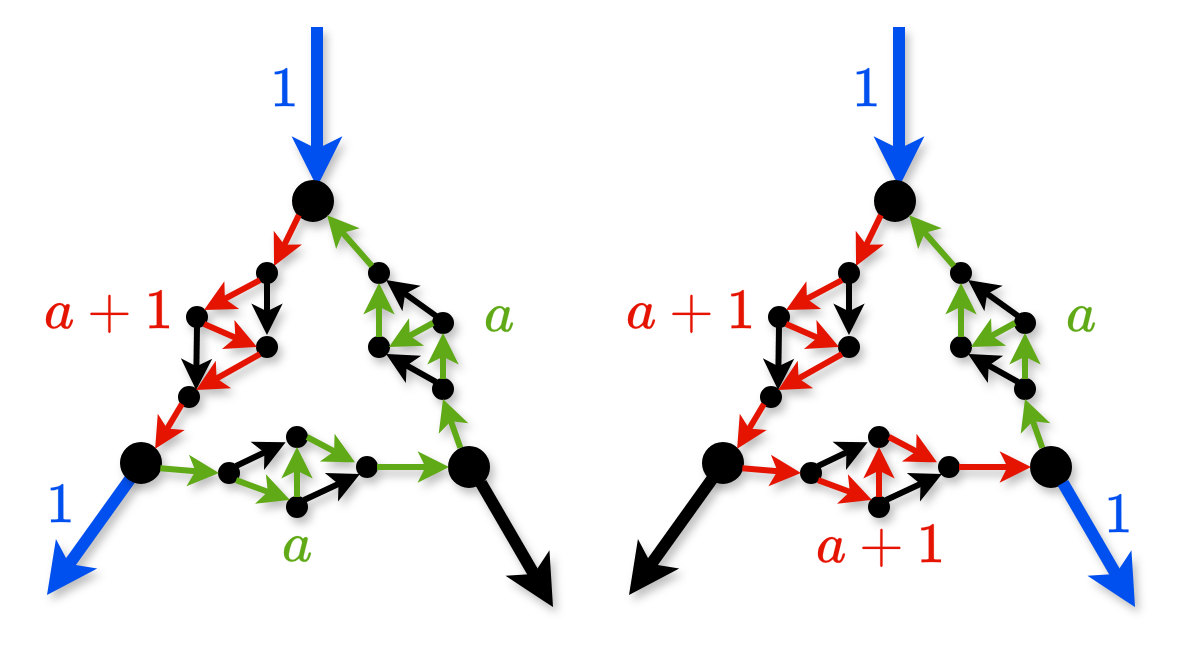}
        \caption{A vertex gadget. Note that the specific structure of the bridge gadgets is just to preserve $3$-regularity; what's important is that they are acyclic. The labelings shown here are the two possible $[a,a+1]\hHAM$ traversals --- in a similar manner, any way to enter and leave the gadget at most $k$ times can be turned into an $[a, a+k]\hHAM$ traversal.}
        \label{fig:3-di-hard-level1}
    \end{figure}
We claim that $G$ is $[1,k]\hHAM$ if and only if $G'$ is $[a, a+k]\hHAM$. 
First we show that if $G'$ is $[a,a+k]\hHAM$ then $G$ must be $[1,k]\hHAM$.
Each
of the bridge gadgets must be visited at least $a$ times, contributing at least
$a$ of the in-label and out-label of each of the hub vertices. Thus, the edges
entering and leaving the vertex gadget can be traversed at most $k+a-a$ times.
That is, each vertex gadget is visited at most $k$ times and at least once in an $[a,a+k]\hHAM$ tour of $G'$. If this is possible, then $G$ must be $[1,k]\hHAM$.

For the other direction, if $G$ is $[1,k]\hHAM$ we can create labelling of $G'$
proving that $G'$ is $[a,a+k]\hHAM$ as follows.
First, copy the labels on the edges in $G$ to the edges that enter and exit the vertex gadgets. 
Now, label the edges internal to the vertex gadgets as follows:
start by assigning all edges in the gadget label $a$. Then, increase the edges
in the first bridge gadget by one for every traversal of the first edge out of the vertex in $G$. 
Then, increase the labels of edges in both the first and second bridge gadgets by one for each traversal of the second edge leaving the vertex in $G$. This labeling visits each interval vertex at least $a$ times and at most $a+k$ times, as desired. Thus, $G'$ is $[a,a+k]\hHAM$.
\end{proof}

Now, by applying this construction recursively, we can show our full hardness result:

\begin{lemma}\label{lem:3dirhard}
    For any $1 \leq a < b$, $[a,b]\hHAM$ is $\NP$-hard in $3$-regular digraphs.
\end{lemma}
\begin{proof}
We will argue by induction on $k$ that for all constants $a\ge 1, k\ge 1$, 
$[a,a+k]\hHAM$ is \NP-hard.
Using the classic theorem that $[1,1]\hHAM$ is \NP-hard in $3$-regular digraphs, 
by \cref{lem:3regbasecase} we deduce $[a,a+1]\hHAM$ is \NP-hard for all $a\ge
1$; this will serve as the base case of our induction.
Now, assume that $[a,a+k]\hHAM$ is \NP-hard for any constant $a$.
In particular this means that $[1,1+k]$ is \NP-hard. 
Then, applying \cref{lem:3regbasecase} we conclude that 
$[a,a+k+1]\hHAM$ is \NP-hard for all $a\ge 1$.
\end{proof}

\cref{lem:3direz} and \cref{lem:3dirhard} together establish \cref{3digraph}:

\begin{theorem}\label{3digraph}
    In directed graphs of maximum degree $3$, for $1 \leq a \leq b$, $[a,b]\hHAM \in \P$ if $a = b > 1$, and $\NP$-hard otherwise. 
\end{theorem}

\subsubsection{Max degree 4}

The hardness for max degree $3$ digraphs immediately gives that $[a, b]\hHAM$ is
hard in max degree $4$ digraphs as long as $b > a$. In this section, we will
take a different reduction approach to show that $[a,a]\hHAM$ is also hard. We
will use the following depletor-like construction, which we call an
\defn{at-least-$k$ gadget}:

\begin{lemma}\label{lem:atleastoneedge}
     For any $k\in [1,a]$, there exists a gadget of max degree $4$ with one incoming edge and one outgoing edge that can be part of an $[a,a]\hHAM$ cycle as long as it's traversed at least $k$ times and at most $a$ times.
\end{lemma}

\begin{proof}
     First, as a base case, consider the construction in \cref{fig:atleastoneedge}. It is clear that this is an at-least-$1$ gadget.

     \begin{figure}[htb]
         \centering
         \includegraphics[width=0.25\linewidth]{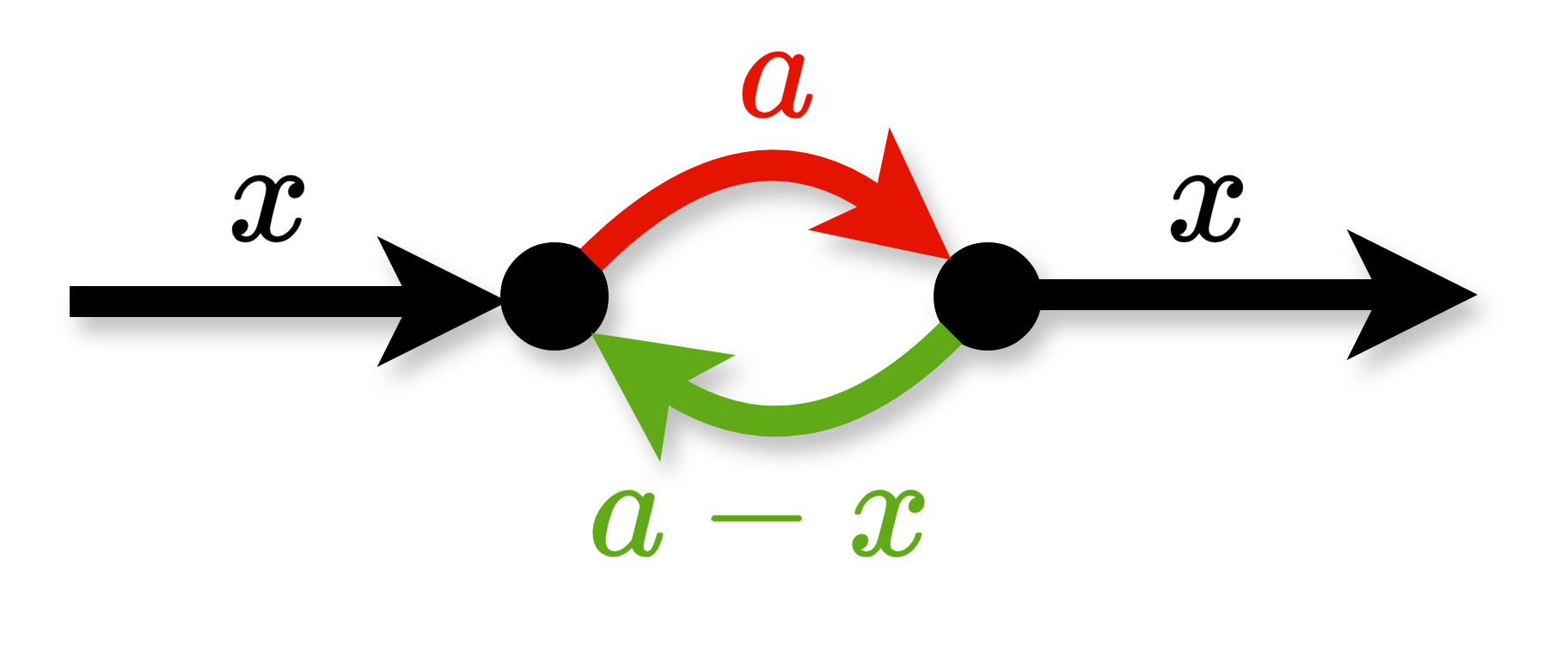}
         \caption{A gadget that can be traversed between $1$ and $a$ times.}
         \label{fig:atleastoneedge}
     \end{figure}
     
    Now, suppose we know how to build an at-least-$(k-1)$ gadget. Then, the construction in \cref{fig:atleastk} is an at-least-$k$ gadget.

    \begin{figure}[htb]
         \centering
         \includegraphics[width=0.45\linewidth]{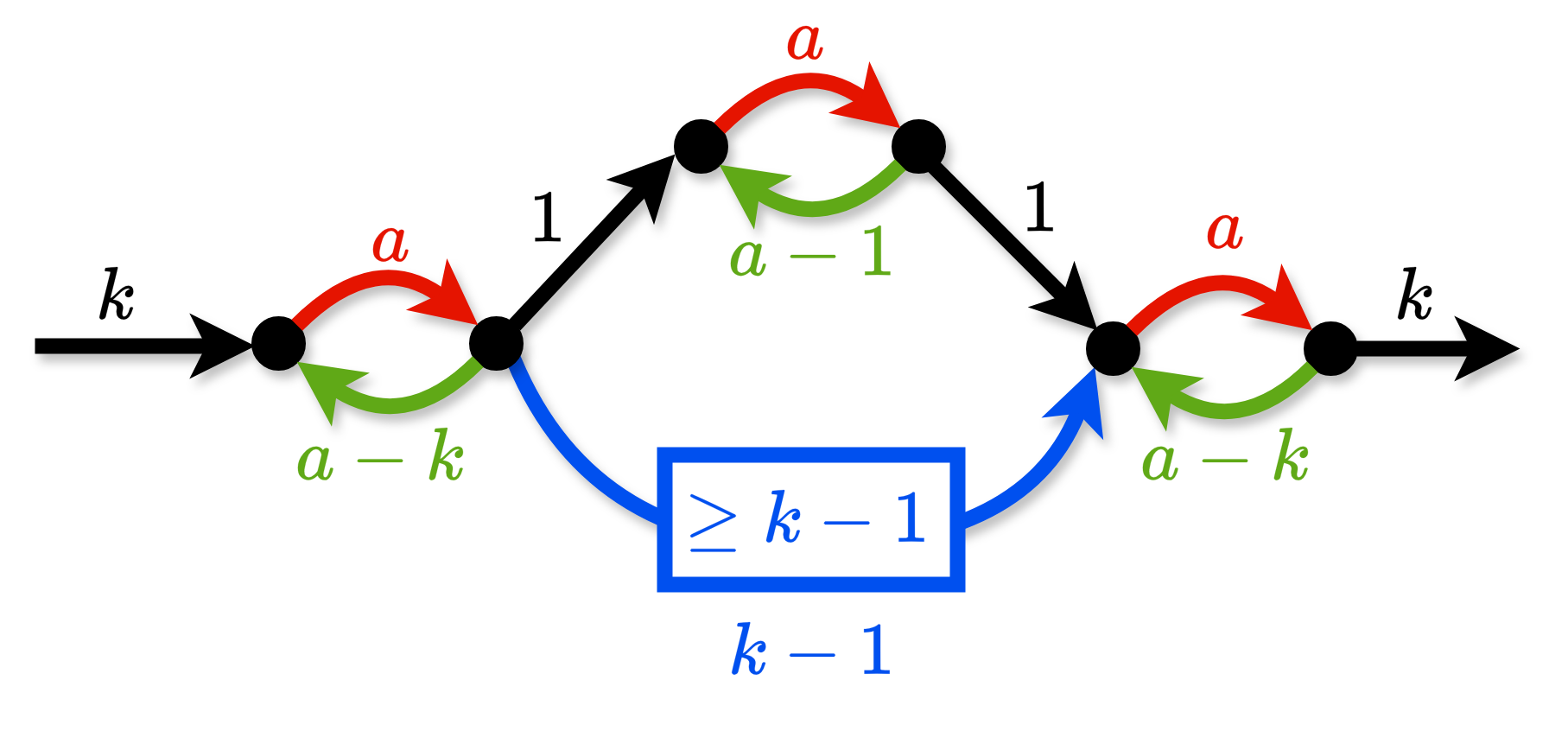}
         \caption{A gadget that can be traversed between $k$ and $a$ times.}
         \label{fig:atleastk}
     \end{figure}
     
\end{proof}

By using at-least-$(a-1)$ gadgets, we can give our desired reduction. 

\begin{lemma}\label{lem:4hard}
    For any $a \geq 1$, $[a,a]\hHAM$ is $\NP$-hard in directed graphs of maximum degree $4$.
\end{lemma}
\newcommand{\vin}{\text{in}}
\newcommand{\vout}{\text{out}}

\begin{proof}
    We reduce from $[1,1]\hHAM$ in max degree $3$ digraphs. Let $G$ be a max degree $3$ digraph. We will make two copies, $V_{\vin}$ and $V_{\vout}$, of the vertex set, and add the edge $(u_{\vout}, v_{\vin})$ for every edge $(u, v) \in E(G)$. We will also add the edges $(v_{\vin}, v_{\vout})$ for all $v \in V(G)$. Finally, for every $v \in V(G)$, we will add an at-least-$(a-1)$ gadget from $v_{\vout}$ to $v_{\vin}$. This construction is shown in \cref{fig:4hard}.

    \begin{figure}[H]
        \centering
        \includegraphics[width = 0.45\linewidth]{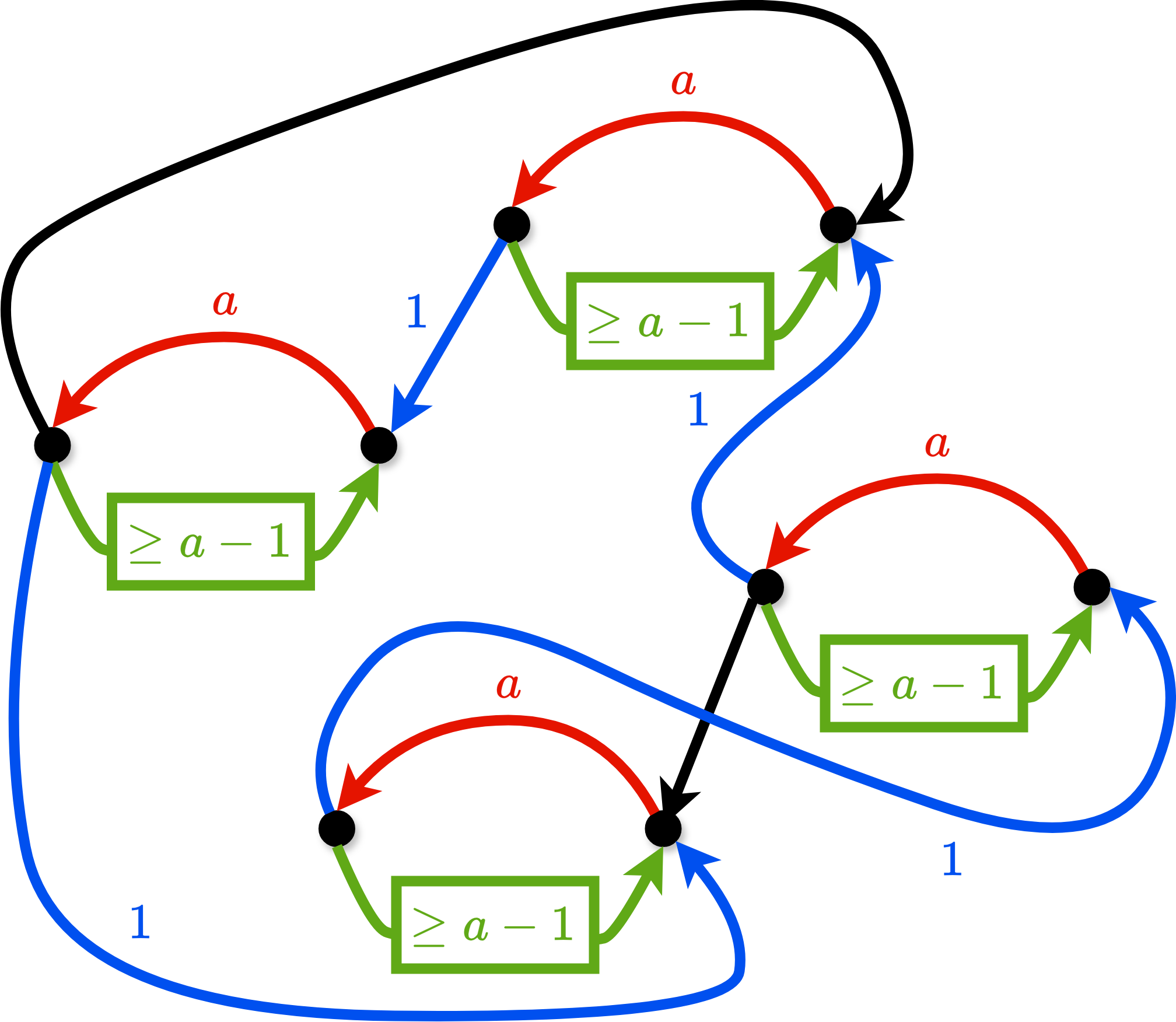}
        \caption{Reducing $[1,1]\hHAM$ in max degree $3$ digraphs to $[a,a]\hHAM$ in max degree $4$ digraphs.}
        \label{fig:4hard}
    \end{figure}

    The at-least-$(a-1)$ gadget consumes all but $1$ of $v_{\vout}$'s available out-label, and all but $1$ of $v_{\vin}$'s available in-label, so at most one other edge out of $v_{\vout}$ can be given nonzero label, and at most one other edge into $v_{\vin}$ can be given nonzero label. Thus, in order to maintain connectivity between all $v$, the labeling in $G'$ must correspond to a Hamiltonian cycle in $G$.
\end{proof}

This establishes \cref{4digraph}:

\begin{theorem}\label{4digraph}
    For all $b \geq a \geq 1$, $[a,b]\hHAM$ is $\NP$-hard in directed graphs of maximum degree 4. 
\end{theorem}

\section{Conclusion and open questions}

We now have full characterizations of which $a$ and $b$ make $[a,b]\hHAM$
\NP-hard in $d$-regular graphs and graphs of max degree $d$. One question of
further potential interest would be to extend this analysis to low-dimensional
grid graphs (i.e. induced subgraphs of a square lattice). It's known that
$[1,1]\hHAM$ is \NP-hard in 2-dimensional grid graphs \cite{gridgraphs}, but not
clear whether one should expect $[1,2]\hHAM$ to be hard as well, for example.
Another problem worth considering would be hardness of approximation --- for
instance, when is it possible to distinguish between $[a,a]\hHAM$ graphs and
graphs that aren't even $[a',a']\hHAM$, for some $a' > a$?

\section*{Acknowledgments}

This paper was initiated during open problem solving in the MIT class
``Algorithmic Lower Bounds: Fun with Hardness Proofs'' (6.5440)
taught by Erik Demaine in Fall 2023.
We thank the other participants of that class for helpful discussions
and providing an inspiring atmosphere.

\printbibliography

\end{document}